\documentclass[11pt, letterpaper]{amsart}
\usepackage{amsmath,amsthm}
\usepackage{amssymb,graphicx,float,bbm,amssymb,mathtools,epsfig,epstopdf}
\usepackage[a4paper]{geometry}
\usepackage{setspace, color}

\usepackage{enumerate}
\usepackage{subfig}
\usepackage{bm}
\usepackage{pstricks}
\makeatother
\newtheorem{theorem}{Theorem}[section]
\newtheorem{lemma}[theorem]{Lemma}
\newtheorem{proposition}[theorem]{Proposition}

\usepackage{amsmath}
\DeclareMathOperator*{\argmax}{arg\,max}

\usepackage[colorlinks=true,citecolor=blue]{hyperref}

\usepackage{indentfirst}

\theoremstyle{remark}
\newtheorem{remark}[theorem]{Remark}
\theoremstyle{definition}
\newtheorem{definition}[theorem]{Definition}
\theoremstyle{assumption}
\newtheorem{assumption}[theorem]{Assumption}
\usepackage{natbib}

\usepackage{natbib}

\theoremstyle{definition}
\newtheorem{example}[theorem]{Example}

\newcommand{\rev}[1]{{\color{black} #1}}

\begin{document}
\title[Nash Equilibria in Optimal Reinsurance Bargaining]{Nash Equilibria in Optimal Reinsurance Bargaining}

\author{Michail Anthropelos}
\address{Department of Banking and Financial Management\\
University of Piraeus\\
Piraeus, Greece}
\email{anthropel@unipi.gr}
\thanks{M. Anthropelos is supported in part by the Research Center of the University of Piraeus}

\author{Tim J. Boonen}
\address{Amsterdam School of Economics\\
University of Amsterdam\\
Amsterdam, The Netherlands}
\email{T.J.Boonen@uva.nl}
\thanks{\rev{We thank two anonymous reviewers, Yichun Chi and Fangda Liu for their valuable comments.}}

\maketitle

\begin{abstract}

We introduce a strategic behavior in reinsurance bilateral transactions, where agents choose the risk preferences they will appear to have in the transaction. Within a wide class of risk measures, we identify agents' strategic choices to a range of risk aversion coefficients. It is shown that at the strictly beneficial Nash equilibria, agents appear homogeneous with respect to their risk preferences. While the game does not cause any loss of total welfare gain, its allocation between agents is heavily affected by the agents' strategic behavior. This allocation is reflected in the reinsurance premium, while the insurance indemnity remains the same in all strictly beneficial Nash equilibria. Furthermore, the effect of agents' bargaining power vanishes through the game procedure and the agent who gets more welfare gain is the one who has an advantage in choosing the common risk aversion at the equilibrium.
\smallskip

\noindent\begin{bfseries}Key-words\end{bfseries}: optimal reinsurance contract; Nash bargaining; Nash equilibrium; strategically chosen risk aversion; risk-sharing games.
\end{abstract}
\section{Introduction}
This paper proposes a game-theoretic bargaining \rev{approach} in optimal reinsurance and the strategic behavior of the insurer and reinsurer, where the strategic set refers to a risk preference parameterization. In particular, the insurer and the reinsurer strategically choose the risk preferences that will apply to the reinsurance transaction. For any submitted risk preferences, a reinsurance contract is given by an indemnity function and a premium. This paper particularly shows that the Nash equilibrium occurs when one agent (insurer or reinsurer) mimics the preferences of the other agent, which establishes an \textit{endogenous homogeneity} on the effective risk preferences.

Originating from the seminal work of \cite{borch1960} and \cite{arrow1963}, the optimal reinsurance problem is traditionally studied as the optimization of the \rev{expected} utility of the insurer given a premium principle. A premium principle is used to represent the interests of the reinsurer, and popular examples are the expected value principle, the standard deviation principle, and the exponential principle (see, e.g., \cite{kaas2008} for an overview). \cite{raviv1979}, \cite{aase2009}, and \cite{boonen2016} propose a different perspective, where the reinsurer is modelled as an agent that bargains with the insurer. We follow these approaches and consider preferences that are comonotonic additive. Such preferences are first characterized by \cite{schmeidler1986}, while used in optimal reinsurance \rev{problems} by \cite{boonen2016}.

In the formulation of the game, we propose that for any submitted risk preferences by the agents, the \textit{asymmetric Nash bargaining solution} produces a reinsurance contract. The asymmetric Nash bargaining solution can be seen as a generalization of the Nash bargaining solution, which is a seminal solution concept in cooperative bargaining \citep{nash1950}. \rev{In simplified terms, asymmetric Nash bargaining models the way the insurer and the reinsurer are going to share the gain that is created by the reinsurance transaction. In fact, within this concept, using a simple parametrization, we may also impose agents with different bargaining power.}

There have been many papers that have studied the Nash bargaining solution \citep[see, e.g.,][]{vandamme1986,rubinstein1992,britz2010}, and it is arguably the most popular solution concept in cooperative bargaining. \cite{kalai1977} shows that such asymmetric Nash bargaining solutions are Pareto optimal and individually rational.\footnote{\rev{A reinsurance contract is called Pareto optimal when there is no other reinsurer contract that is better for both agents and strictly better for at least one agent. In actuarial science and optimal reinsurance contract theory, several notions of Pareto optimality have been studied by \cite{boonen2016,cai2017,asimit2018,lo2019}.}} \cite{asimit2018} show that Pareto optimality and individual rationality constitute reinsurance contracts where the indemnity function minimizes a particular sum of risk measures (while the corresponding premium is flexible and not uniquely determined). The asymmetric Nash bargaining solution then yields a particular choice of this premium.

We readily get that any Pareto-optimal reinsurance contract and the induced welfare gains strongly depend on the agents' risk preferences. It is then reasonable to assume that agents have motive to strategically\footnote{\rev{Following the standard terminology in game theory, we call an agent (or his behavior) \textit{strategic}, if his optimization criterion takes into account the fact that his actions impact the outcome of the transaction (equilibrium). In turn, the set in which the control variable belongs is called \textit{set of strategies or set of strategic choices.}. In our game, the set of strategic choices is the set within which the agents' risk aversion parameter lies.}} choose the risk preferences that will appear to have in the transaction and may state different preferences than their true ones. The rationale behind this argument is that reinsurance transaction is bilateral (i.e., similar to a duopoly) and hence both of the agents can heavily influence the transaction with their actions \rev{(in other words, both agents possess market power)}. This means that agents should take into account their ability to affect the transaction when they negotiate the reinsurance contract, whose most crucial part is apparently the agents' risk preferences. Based on this, we propose a game where each agent's set of strategies is the actual risk preferences that he will apply to the transaction, or in other words, we ask \rev{how} risk averse he is going to appear.

Our focus is on the constrained problem of the optimal reinsurance transaction, where the loss function and the premium result from a bilateral cooperative bargaining. We consider the family of comonotonic additive risk measures for both insurer and reinsurer \citep[as in][]{boonen2016}, and we use a parameterization in a way where the parameter can be seen as the level of risk-aversion. Consistent with this parameterization are the majority of the well-known risk measures.

Following the aforementioned argument, we define as agents' strategic set the interval of admissible risk-aversion parameters, allowing the agents to choose how much risk averse they appear in the risk-sharing. For the proposed game, we further assume that for any submitted risk aversion, the corresponding reinsurance contract will be determined by the optimal sharing rule (i.e., the one that minimizes the sum of risk measures). The justification of this setting is based on the socially optimal welfare gain, in the sense that the loss of utility caused by agents strategic behavior is minimized \citep[also used in similar models of][]{anthropelos2017, Anth17}. In fact, as it is shown in Section \ref{sec:game}, our proposed game does not decrease the total welfare, but rather affects the reinsurance premium and hence the gain's allocation.

\rev{Assuming that both agents apply the same aforementioned strategic behavior, we define as Nash equilibria the pairs of agents' submitted risk aversions at which none of the agents wants to deviate (i.e., a fixed point of agents' optimization problems).} It is shown (Theorem \ref{th:NE}) that within the set of Nash equilibria that result in a strict improvement compared to the status quo, agents' \rev{strategies} is to mimic the risk aversion of his counterparty. In other words, at the equilibrium both agents appear with the same risk aversion (even if their true risk preferences are different). It is also endogenously derived that the possibly different exogenous agents' bargaining power does not influence neither the equilibrium reinsurance contract nor the premium. This is an important feature of our model, since the determination of the agent who has more market power is not exogenously given, but rather it is a part of Nash equilibrium.

The set of Pareto optimal Nash equilibria is not a single-valued and there should be an additional criterion for the selection of a specific one. An idea is based on the so-called Stackelberg equilibrium argument. In this concept, one of the agent \rev{(called the leader)} first states his parameter, while the other \rev{(called the follower)} follows. We show in Theorem \ref{th:Stack} that the Stackelberg equilibrium is indeed a Nash equilibrium, and that all gains from reinsurance transaction (if any) go to the leader by mimicking the preferences of the follower (leaving the follower at \rev{an} indifferent level).

Nash equilibria have been recently studied within the concept of thin financial markets \citep[see, e.g.,][]{AnKarVic19,MalRos17,RosWer08,RosWer15}. Agents (i.e., traders) behave strategically when submitting their demand functions on a given vector of tradeable assets. For instance in \cite{AnKarVic19}, \rev{similar} to our model, traders' set of strategic choices is parameterized by risk aversion coefficient that appears in their demand function. On the other hand, strategies in risk-sharing transactions when the contract is endogenously derived in the equilibrium are recently studied by \cite{anthropelos2017} and \cite{Anth17}. The former considers exponential utility maximizers whose strategic set is the subjective probability measure that agents declare in the transaction; while in \cite{Anth17} each agent strategically chooses the risky portfolio he is willing to share with the other agent.

Stackelberg equilibria have been recently studied for the context of price competition in a duopoly by \cite{albrecher2017}, while \cite{chen2018} study them in dynamic differential games. The closest approaches to ours are the ones of \cite{chan1985} and \cite{cheung2019}, who study Bowley solutions. In Bowley solutions, the reinsurer serves as the leader. Then, the risk-neutral reinsurer states a Wang's premium principle in such a way that it anticipates the optimal reinsurance contract selected by the insurer given this premium principle. In our approach however, the total premium is determined by a Nash bargaining solution, which is a result of bargaining under the strategically chosen risk preferences.
\medskip

This paper is set out as follows. Section \ref{sec:model} states the bargaining model without strategic behavior. Section \ref{sec:game} introduces the game where the insurer and reinsurer are allowed to choose the risk preferences in the transaction, while Section \ref{sec:NE} provides a characterization of set of the Nash equilibria and the exact structure the form of the related Stackelberg equilibria. Finally, Section \ref{sec:conclusion} concludes.

\section{The Model}\label{sec:model}
We study the case of an insurer \rev{seeking} reinsurance. The insurer is endowed with initial risk (his insurance portfolio) $X$ that is realized at a given future reference time. We assume that $X$ is a random variable defined on a probability space $(\Omega,\mathcal{F}, \mathbb{P})$ and  $X\in L^\infty$, i.e.~$X$ belongs to the class of bounded random variables on $(\Omega,\mathcal{F}, \mathbb{P})$. As usually in the literature, we also assume that $X$ is non-negative and not a constant almost surely.

The insurer can cede part of the risk $X$ to the reinsurer. In particular, the insurer will cede $I(X)$ to the reinsurer, and the risk $X-I(X)$ is retained, where function $I$ determines the reinsurance payment. Naturally, to interpret $I(X)$ as an indemnity, it holds that $0\leq I(X)\leq X$ and also assume that $I\in\mathcal I$, where
\begin{align}
  \ \mathcal{I}:= \{I(\cdot): I(0)=0, \ 0\leqslant I(x)-I(y)\leqslant x-y, \ \forall \  0\leq y<x\}.
\end{align}
In other words, the loss functions $I(x)$ and $x-I(x)$ are increasing and any increment in compensation is always less than or equal to the increment in loss. Note that these properties of the ceded loss function $I$ are linked to the discouragement of  moral hazard \citep{denuit1998,young1999}. In fact, $I\in\mathcal I$ is equivalent to $I(0)=0$, $I$ is absolutely continuous, and $0\leq I'(x)\leq 1$ for all $x\geq 0$, almost everywhere \citep{zhuang2016}.

We index the insurer as agent 1 and the reinsurer as agent 2. The insurer and reinsurer are both endowed with risk measures that they aim to minimize. These are denoted by $\rho_1$ and $\rho_2$ and satisfy the following assumption.
\begin{assumption}\label{ass:1} For $i\in\{1,2\}$, the risk measure $\rho_i$ satisfies (see also the related discussion in \cite{boonen2016}):
\begin{itemize}\item monotonicity with respect to the order of $L^\infty$; \item $\rho_i(0)=0$ and $\rho_i(1)=1$;
\item comonotonic additivity: $\rho_i(Y)=\rho_i(I(Y))+\rho_i(Y-I(Y))$ for all $Y\in L^\infty$ and all $I\in \mathcal I$.\end{itemize}
\end{assumption}

Note that the normalization and comonotonic additivity imply that $\rho_i$ is cash-invariant, i.e., $\rho_i(Y+a)=\rho_i(Y)+a$ for every $Y\in L^\infty$ and $a\in\mathbb{R}$. Also, comonotonic additivity and monotonicity together with a regularity assumption on continuity imply the Choquet representation of
\cite{schmeidler1986}. If $\rho_i$ is also law-invariant\footnote{A risk measure $\rho_i$ is called law-invariant if for any $Y_1,Y_2\in L^\infty$ with $Y_1\stackrel{d}{=}Y_2$, then $\rho_i\big(Y_1\big)=\rho_i\big(Y_2\big)$.}, it can be represented as a distortion risk measure \citep{wang1997}. Risk measure $\rho_i$ is a distortion risk measure when:
\begin{equation*}\rho_i(Y)=\int_0^\infty g_i(S_Y(z))dz+\int^0_{-\infty}[ g_i(S_Y(z))-1]dz, \text{ for all } Y\in L^\infty,\end{equation*} where $S_Y(z):=1-F_Y(z)$ is the survival function of $Y$ and the distortion function $g_i$ is non-decreasing, $g_i(0)=0$, and $g_i(1)=1$.
Popular examples of such distortion risk measures are the Value-at-Risk (VaR) and the Conditional Value-at-Risk (CVaR).

In return for the coverage of the indemnity $I(X)$, the insurer pays a non-negative premium $\pi\geq0$ to the reinsurer. The insurer aims to minimize $\rho_1(X-I(X)+\pi)$, and the reinsurer aims to minimize $\rho_2(I(X)-\pi)$ for a reinsurance contract $(I,\pi)\in\mathcal I\times \mathbb{R}_+$.

As a first economic criterion, we propose individual rationality. A contract $(\hat{I},\hat{\pi})\in\mathcal I\times \mathbb{R}_+$ is called \textit{individually rational} if
\begin{align}\label{eq:IR1}\rho_1(X-\hat I(X)+\hat \pi)&\leq \rho_1(X),\\ \label{eq:IR2}\rho_2(\hat I(X)-\hat\pi)&\leq \rho_2(0)=0.\end{align}
The individual rationality \rev{is} obviously a minimal requirement (since if it does not hold then agents will not agree to the reinsurance transaction). The status quo situation (i.e., when agents do not proceed to any reinsurance) is given by $I(X)\stackrel{d}{=}0$ and $\pi=0$.

As a second economic criterion we impose the well-known \textit{Pareto optimality}. A contract $(\hat{I},\hat{\pi})$ is called Pareto optimal when there \rev{is} no other contract $(I,\pi)\in\mathcal I\times \mathbb{R}_+$ such that $\rho_1(X-I(X)+ \pi)\leq \rho_1(X-\hat{I}(X)+\hat{\pi})$ and $\rho_2( I(X)-\pi)\leq \rho_2(\hat{I}(X)-\hat{\pi})$, with one inequality strict.

According to Theorem 3.1 of \cite{asimit2018}, an individually rational reinsurance contract $(\hat{I},\hat{\pi})$ is Pareto optimal if and only if \rev{$\hat{I}$} solves
\begin{equation}\label{eq:I}\min_{I\in\mathcal I}\rho_1(X-I(X))+\rho_2(I(X)).\end{equation}

It is important to emphasize at this point that (due to the cash-invariance property) the determination of the Pareto-optimal contract does not include the premium. This means that there should be another criterion to determine the premium for a Pareto-optimal loss function $I$. The criterion that we impose for this is given by the \textit{asymmetric Nash bargaining} solution. A characterization of the asymmetric Nash bargaining is provided by \cite{kalai1977} and is similar as the Nash bargaining solution (originally introduced in \cite{nash1950}), but without a \emph{symmetry} axiom. More precisely, the set of asymmetric Nash bargaining solutions, with bargaining power $\delta\in(0,1)$ for reinsurer and $1-\delta$ for the insurer, is given by
\begin{align}\label{eq:ANBS}
\nonumber\max_{(I,\pi)\in\mathcal I\times\mathbb{R}_+} & (\rho_1(X)-\rho_1(X-I(X)+\pi))^{1-\delta}(-\rho_2(I(X)-\pi))^\delta,\\
\text{s.t.}\,\, &\rho_1(X)\geq\rho_1(X-I(X)+\pi),\\
\nonumber & 0 \geq\rho_2(I(X)-\pi),\end{align} which does not need to be single-valued.
Note that, as in the seminal paper of \cite{kalai1977}, the above characterization is for the convex problem. In general, a convex bargaining problem is determined by a convex and compact set $A \subset\mathbb{R}^2$ of feasible utility levels and a disagreement point $d\in\mathbb{R}^2$. In our case, the utility level is the negative of a risk measure and the disagreement point is the vector $d = (-\rho_1(X),0)$. Again by \citep{kalai1977} we have that the asymmetric Nash bargaining solution is individually rational and Pareto optimal, and thus the indemnity contract $I$ solves \eqref{eq:I}. Moreover, by \cite{boonen2016}  we have that any asymmetric Nash bargaining solution $(I,\pi)$ is such that $I$ solves \eqref{eq:I} and
\begin{equation}\label{e.price}
\pi=\rho_2(I(X))+\delta\Big(\rho_1(X)-\rho_1(X-I(X))-\rho_2(I(X))\Big).
\end{equation}
Thus, according to the third imposed criterion, for a given Pareto-optimal $I$, the associate premium is induced by \eqref{e.price}. Throughout the rest of this paper, we keep the parameter $\delta\in(0,1)$ fixed. Note that for all asymmetric Nash bargaining solutions $(I,\pi)$, the vector of risk measures $(\rho_1(X-I(X)+\pi), \rho_2(I(X)-\pi))$ remains the same.

We may identify $WG(I(X)):=\rho_1(X)-\rho_1(X-I(X))-\rho_2(I(X))$ as the \textit{total welfare gain} from trading. 
\rev{Since $I$ solves \eqref{eq:I} and $I(X)=0$ is an element of $\mathcal I$, it follows directly that $WG(I(X))\geq0$.}
The fraction $\delta$ of it is the welfare gain for the reinsurer, and the remaining faction $1-\delta$ of this total welfare gain is the welfare gain for the insurer. To formalize this, it holds for $\pi$ in \eqref{e.price} that
\begin{align*}
\rho_1(X)-\rho_1(X- I(X)+\pi)&=(1-\delta)\left(\rho_1(X)-\rho_1(X-I(X))-\rho_2(I(X))\right)=(1-\delta)WG(I(X)),\\
\rho_2(0)-\rho_2(I(X)-\pi)&=\delta\left(\rho_1(X)-\rho_1(X-I(X))-\rho_2(I(X))\right)=\delta WG(I(X)),\end{align*}
where $I\in\mathcal I$. Here, we use cash-invariance of $\rho_1$ and $\rho_2$. Note also that $\pi=\rho_2(I(X))+\delta WG(I(X))$, and when $\delta\rightarrow0$ or $\delta\rightarrow1$ then the indifference premiums of the reinsurer and insurer are charged, respectively.

\begin{remark}\label{r.delta}
It is apparent from the definition of the asymmetric Nash bargaining solution, that the parameter $\delta\in(0,1)$ denotes the exogenously given bargaining power of the agents. Since the Pareto-optimal loss function $I$ solves problem \eqref{eq:I}, the bargaining power affects only the level of the associated premium. In other words, we get directly from \eqref{e.price}, that higher reinsurer's bargaining power (i.e., the level of $\delta$) means higher premium for each optimal $I$. This is directly linked to the sharing of the welfare gain between the insurer and the reinsurer, which is based on the exogenously imposed parameter $\delta$.
\end{remark}

\section{Preferences as Strategic Choices}\label{sec:game}

In this section, we develop the argument on the proposed strategic behavior of the agents. As mentioned in the introductory section, the main idea is that agents do have motive to strategically choose the risk preferences (i.e., their risk measures) that they will submit to the transaction. As we have seen in Section \ref{sec:model}, both the reinsurance loss function $I$ and its premium $\pi$ heavily depends on the both agents' risk measures. In turn, the total welfare and its sharing between the agents hang upon the risk measures that each agent declares in the transaction. Furthermore, each of the agents possesses (possibly not symmetric) market power since the reinsurance contract can be seen as a specific duopoly zero-supply transaction. Therefore, it is reasonable to assume that both agents act strategically regarding the risk measure they appear to have when determining the reinsurance contract. In this way, a Nash game is formed and its equilibrium point will induce the effective risk measures and the associate reinsurance contract.

For any risk measure that agents declare, the reinsurance contract will be of the form of the asymmetric Nash bargaining solution, as defined in \eqref{eq:ANBS}. More precisely, we assume that whatever risk preferences are stated, the reinsurance loss function will be designed by the Pareto-optimal rule, while the induced premium \rev{will keep} the same bargaining power parameter as in \eqref{eq:premium}. In other words, agents pre-agree to share the insurer's risky portfolio in the way that optimizes the total welfare for each strategically stated risk measures. In terms of total welfare concerns, this is a reasonable framework, which in fact has been used in similar equilibrium models; see among others \cite{anthropelos2017, Anth17} and the references therein.

In short, we propose the following ``game'':
\begin{itemize}
\item [1.] the two agents state their risk measure $\rho^*_i$ (possibly different than $\rho_i$);
\item [2.] given the pair $(\rho_1^*,\rho^*_2)$, the asymmetric Nash bargaining solution selects an optimal reinsurance contract;
\item [3.] if this asymmetric Nash bargaining solution is not individually rational for at least one agent, there will be no reinsurance.
\end{itemize}

The next step is to parameterize set of admissible risk measures, that is to parameterize the agents' strategic sets. For this, we impose the following parameterization of $\rho_i$ for both $i\in\{1,2\}$:
\begin{equation*}\rho_i(Y)=\rho(Y;\gamma_i),\text{ for all }Y\in L^\infty.\end{equation*}
Throughout this paper, we assume that $(\gamma_1,\gamma_2)$ is the couple of \emph{true} parameters of the two agents. Therefore, the agents'(true) risk measures $\rho_1$ and $\rho_2$ belong to the same class of risk measures with possibly different values of parameters $\gamma_1$ and $\gamma_2$.

Imposing the above family of risk measures implies that agents are allowed to strategically choose the level of the parameter $\gamma_i$ by a value $\zeta_i$. We assume that the domain of $\zeta_i$, that is the agent's \textit{strategic set}, is a closed and finite interval, normalized to $[0,1]$.
In order to get a further structure on the risk measure parameterization, we impose the following properties on $\rho$.
\begin{assumption}\label{ass:M}For all $I\in\mathcal I$ such that $\mathbb{P}(I(X)>0)>0$, it holds that $\rho(I(X);\cdot)$ is continuous, and strictly increasing on $[0,1]$. Moreover, for all $\gamma\in[0,1]$, $\rho(\cdot;\gamma)$ satisfies Assumption \ref{ass:1}. \end{assumption}

Intuitively, increasing with respect to $\gamma$ implies that parameterization could stand for risk-aversion coefficient, in the sense that the agent's risk measure becomes more conservative as $\gamma$ increases. Besides that, the other good features of the chosen parameterization are its tractability and the fact that it is consistent with the most commonly-used examples of risk measures:
\begin{example}\label{ex:1} Let $\rho(\cdot;\gamma)$ be the distortion risk measure parameterized by distortion function $g(\cdot;\gamma)$. If $g(s;\gamma)$ is strictly increasing in $\gamma$ for all $s\in(0,1)$, then the distortion risk measure $\rho$ satisfies Assumption \ref{ass:M}. We state two different examples of such distortion risk measures. 
First, let $\rho$ be defined as:
\begin{align}\rho(Y;\gamma)&=  (1-\gamma) E[Y]+\gamma \hat\rho(X),\end{align}
where $\hat\rho$ is a distortion measure with a non-linear, concave distortion function $\hat g$. Thus, it holds that $\hat g(s)>s$ for all $s\in(0,1)$. Then, $\rho(\cdot;\gamma)$ is a distortion risk measure with distortion function \rev{$g(s;\gamma)=(1-\gamma)s+\gamma\hat g(s)$} for $s\in[0,1]$, which is strictly increasing in $\gamma$ on $(0,1)$. For instance, we could assume that $\hat\rho(Y)= CVaR_{\alpha}(Y)$ for some $\alpha\in(0,1)$, that is the distortion risk measure with distortion function $g(s)=\min\{s/(1-\alpha),1\}$ \citep[see, e.g.,][]{dhaene2006}.

The second example is the proportional hazard transform. This is again a distortion risk measure, where the corresponding distortion function can be scaled such that it satisfies Assumption \ref{ass:M}. For instance, the distortion function is $g(s;\gamma)=s^{1-\gamma}$ for $s\in[0,1]$ \citep{wang1995}.\end{example}

As stated before, for any couple of submitted parameters $(\zeta_1,\zeta_2)\in[0,1]^2$, that is for any submitted risk measures $(\rho(\cdot;\zeta_1),\rho(\cdot;\zeta_2))$, agents select the corresponding Pareto-optimal reinsurance contract. For any pair $(\zeta_1,\zeta_2)\in[0,1]^2$, we shall call a reinsurance contract $(I,\pi)\in\mathcal I\times\mathbb{R}_+$, $(\zeta_1,\zeta_2)$-Pareto optimal when $I$ solves
\begin{equation}\label{eq:I2}\min_{I\in\mathcal I}\rho(X-I(X);\zeta_1)+\rho(I(X);\zeta_2),\end{equation} and $\pi\geq0$ is chosen freely.

The following proposition characterizes the $(\zeta_1,\zeta_2)$-Pareto optimal reinsurance contracts within the proposed parameterization.
\begin{proposition}\label{prop:PO} Suppose that $\rho$ satisfies Assumption \ref{ass:M}. Then, $(I^*,\pi)$ is $(\zeta_1,\zeta_2)$-Pareto optimal if and only if it holds almost surely that
$$I^*(X)=\left\{\begin{array}{ll}X &\text{ if }\zeta_1>\zeta_2,\\ \hat I(X) &\text{ if }\zeta_1=\zeta_2,\\ 0 & \text{ if }\zeta_1<\zeta_2,\end{array}\right. $$
where $\hat I$ is any function in $\mathcal I$.\end{proposition}
\begin{proof} First, assume that $\zeta_1>\zeta_2$. For all $I\in\mathcal I$ such that $\mathbb{P}(I(X)< X)>0$, we get
\begin{equation*}\rho(X-I(X);\zeta_1)+\rho(I(X);\zeta_2)>\rho(X-I(X);\zeta_2)+\rho(I(X);\zeta_2)=\rho(X;\zeta_2),\end{equation*} where the inequality is due to monotonicity of $\rho(X-I(X);\cdot)$, and the equality is due to comonotonic additivity of $\rho$. On the other hand, if $I(X)=X$ almost surely, then
\begin{align*}\rho(X-I(X);\zeta_1)+\rho(I(X);\zeta_2)=
\rho(X;\zeta_2),\end{align*}since $\rho(0;\zeta_1)=0$. Hence, the minimum in \eqref{eq:I2} is attained if and only if $I(X)=X$ almost surely.

The proof for the case that $\zeta_1<\zeta_2$ is similar. Finally, the result for the case where $\zeta_1=\zeta_2$ is a direct consequence of comonotonic additivity of $\rho(\cdot;\zeta_1)$.
\end{proof}

It follows that under risk measures of the same family, if the reinsurer appears in the transaction as relatively more risk averse, then there is no transfer of risk. On the other hand, if reinsurer behaves less risk averse than the insurer, there is a total risk transfer to the reinsurer. If they declare the same risk aversion, then, for every $I\in\mathcal{I}$, the quantity of problem \eqref{eq:I2} stays the same and equal to $\rho(X;\zeta_1)$.

Regarding the premium, if agents' submitted risk aversions are $(\zeta_1,\zeta_2)$, the asymmetric Nash bargaining solution induces the following premium (see and compare with \eqref{e.price}):
\begin{align}
\pi&=\rho(I^*(X);\zeta_2)+\delta\Big(\rho(X;\zeta_1)-\rho(X-I^*(X);\zeta_1)-\rho(I^*(X);\zeta_2)\Big)\nonumber\\
&=\delta[\rho(X;\zeta_1)-\rho(X-I^*(X);\zeta_1)]+(1-\delta)\rho(I^*(X);\zeta_2),\label{eq:premium}\end{align} for reinsuring the risk $I^*(X)$, where $I^*\in\mathcal I$ is as in Proposition \ref{prop:PO}. For a given pair $(\zeta_1,\zeta_2)$, we refer to $(I^*,\pi)$, given in Proposition \ref{prop:PO} and \eqref{eq:premium}, as the asymmetric Nash bargaining solution. If $\zeta_1<\zeta_2$, then $I^*(X)\stackrel{d}{=}0$ and $\pi=0$, and if $\zeta_1>\zeta_2$, then $I^*(X)\stackrel{d}{=}X$ and $\pi$ given by \eqref{eq:premium}.

The risky position after reinsurance (\emph{posterior} risk) for each agent will then be given by
\begin{align*}
Y_1&=X-I^*(X)+\delta[\rho(X;\zeta_1)-\rho(X-I^*(X);\zeta_1)]+(1-\delta)\rho(I^*(X);\zeta_2),\\ Y_2&=I^*(X)-\delta[\rho(X;\zeta_1)-\rho(X-I^*(X);\zeta_1)]-(1-\delta)\rho(I^*(X);\zeta_2).\end{align*}
So, after the reinsurance contract $(I^*,\pi)$ is traded, the insurer is endowed with $Y_1$ and the reinsurer is endowed with $Y_2$.  Note that the posterior risk is evaluated by the agents using the \emph{true} preferences, that is the measures $\rho(\cdot;\gamma_i)$ are applied.

In other words, for any pair $(\zeta_1,\zeta_2)$, provided that $\zeta_1\neq\zeta_2$, Proposition \ref{prop:PO} and pricing rule \eqref{eq:premium} determine a unique reinsurance contract. We now clarify the case when $\zeta_1=\zeta_2$, which will be proven to be quite interesting in our set-up. Note that when risk aversions are equal, the comonotonic additivity implies that the welfare gains are zero. Proposition \ref{prop:PO} states that any $\hat I(X)$ with $\hat I\in\mathcal I$ yields $(\zeta_1,\zeta_2)$-Pareto optimality. We differentiate two cases. Firstly, we let\footnote{Here, we say that if $\gamma_1< \gamma_2$, then $[\gamma_2,\gamma_1]=\emptyset$.} $\zeta_1=\zeta_2\in[\gamma_2,\gamma_1]$, which thus implies $\gamma_1\geq \gamma_2$.
By monotonicity of $\rho(X-I(X);\cdot)$ and comonotonic additivity of $\rho(\cdot;\zeta_1)$, it holds for all $\hat I\in\mathcal I$ and $\zeta_1=\zeta_2\leq\gamma_1$ that:
\begin{align*}
\rho(Y_1;\gamma_1)&=\rho(X-\hat I(X);\gamma_1)+\delta[\rho(X;\zeta_1)-\rho(X-\hat I(X);\zeta_1)]+(1-\delta)\rho(\hat I(X);\zeta_2)\\&= \rho(X-\hat I(X);\gamma_1)+\delta(\rho(\hat I(X);\zeta_1)+(1-\delta)\rho(\hat I(X);\zeta_2)\\&=\rho(X-\hat I(X);\gamma_1)+\rho(\hat I(X);\zeta_1)\\&\geq\rho(X-\hat I(X);\zeta_1)+\rho(\hat I(X);\zeta_1)\\&=\rho(X;\zeta_1),
\end{align*}
and if $\hat I\in\mathcal I$ and $\zeta_1=\zeta_2\geq\gamma_2$ then
 \begin{align*}\rho(Y_2;\gamma_2)&=\rho(\hat I(X);\gamma_2)-\delta[\rho(X;\zeta_1)-\rho(X-\hat I(X);\zeta_1)]-(1-\delta)\rho(\hat I(X);\zeta_2)\\&=\rho(\hat I(X);\gamma_2)-\rho(\hat I(X);\zeta_2))\\&=\rho(X;\gamma_2)-\rho(X;\zeta_2))-[\rho(X-\hat I(X);\gamma_2)-\rho(X-\hat I(X);\zeta_2))] \\
 &\geq\rho(X;\gamma_2)-\rho(X;\zeta_2).
 \end{align*}
Hence, if $\zeta_1=\zeta_2$ and  also $\gamma_2\leq\zeta_i\leq\gamma_1$, it is optimal for both agents to select the indemnity $\hat I(X)=X$, which means that $Y_1=\rho(X;\zeta_1)$ and also yields the premium $\pi=\rho(X;\zeta_1)$ from \eqref{eq:premium}.

Secondly, we check the case where $\zeta_1=\zeta_2\notin[\gamma_2,\gamma_1]$, where it is easily shown that at least one agent is strictly better off in case $\hat I(X)=0$ (no risk-sharing), which yields the premium $\pi=0$ from \eqref{eq:premium}. Summing up, if $\zeta_1=\zeta_2$, the asymmetric Nash bargaining solution is not single-valued, and then it is optimal to then select $\hat I(X)=X$ if $\zeta_1=\zeta_2\in[\gamma_2,\gamma_1]$ and $\hat I(X)=0$ if $\zeta_1=\zeta_2\notin[\gamma_2,\gamma_1]$. Hence, we pose the following definition of a single-valued solution that always selects an asymmetric Nash bargaining solution.
\begin{definition}\label{c:same_zeta} Suppose that $\rho$ satisfies Assumption \ref{ass:M}.
For any pair $(\zeta_1,\zeta_2)\in[0,1]^2$, the reinsurance contract $(I^*,\pi)$ is such that
\begin{equation}
I^*(X)=\left\{\begin{array}{ll}X &\text{ if }\zeta_1>\zeta_2\text{ or }\zeta_1=\zeta_2\in[\gamma_2,\gamma_1],\\ 0 & \text{ otherwise,}\end{array}\right. \label{eq:I*}\end{equation} and $\pi$ as in \eqref{eq:premium}.
\end{definition}

For $(I^*,\pi)$ as in Definition \ref{c:same_zeta}, welfare gains of the insurer and the reinsurer are:
\begin{align*}\hat{WG}_1(\zeta_1,\zeta_2)&:=\rho(X;\gamma_1)-\rho(Y_1;\gamma_1)\\&=\rho(X;\gamma_1)-
\rho(X-I^*(X);\gamma_1)-\rho(I^*(X);\zeta_2)\\&\,\,\,\,-\delta\Big(\rho(X;\zeta_1)-\rho(X-I^*(X);\zeta_1)-\rho(I^*(X);\zeta_2)\Big),\\
\hat{WG}_2(\zeta_1,\zeta_2)&:=\rho(0;\gamma_2)-\rho(Y_2;\gamma_2)\\&=
-\rho(I^*(X);\gamma_2)+\rho(I^*(X);\zeta_2)+\delta\Big(\rho(X;\zeta_1)-\rho(X-I^*(X);\zeta_1)-\rho(I^*(X);\zeta_2)\Big).\end{align*}
Substituting $I^*$ from \eqref{eq:I*} yields:
\begin{align*}\hat{WG}_1(\zeta_1,\zeta_2)&=\left\{\begin{array}{ll}\rho(X;\gamma_1)- (1-\delta)
\rho(X;\zeta_2)-\delta\rho(X;\zeta_1)\,\,&\text{ if }\zeta_1>\zeta_2\text{ or }\zeta_1=\zeta_2\in[\gamma_2,\gamma_1],\\ 0 &\text{ otherwise,}\end{array}\right. \\
\hat{WG}_2(\zeta_1,\zeta_2)&=\left\{\begin{array}{ll}-\rho(X;\gamma_2)+(1-\delta)
\rho(X;\zeta_2)+\delta\rho(X;\zeta_1)&\text{ if }\zeta_1>\zeta_2\text{ or }\zeta_1=\zeta_2\in[\gamma_2,\gamma_1],\\ 0 &\text{ otherwise.}\end{array}\right.
\end{align*}
The total welfare gains from trading are then given by
\begin{align}\hat{WG}_1(\zeta_1,\zeta_2)+\hat{WG}_2(\zeta_1,\zeta_2)&=\rho(X;\gamma_1)-\rho(X-I^*(X);\gamma_1)-\rho(I^*(X);\gamma_2),\label{eq:totWG}\end{align} where $I^*$ is given in \eqref{eq:I*}. The right hand side of \eqref{eq:totWG} depends on $(\zeta_1,\zeta_2)$ only via $I^*$. Note that \eqref{eq:totWG} is equal to $\rho(X;\gamma_1)-\min_{I\in\mathcal I}\{\rho(X-I(X);\gamma_1)+\rho(I(X);\gamma_2)\}$ when $(\zeta_1,\zeta_2)$ is ordered in the same way as $(\gamma_1,\gamma_2)$. Moreover, if $\gamma_1>\gamma_2$ and $\zeta_1>\zeta_2$ or $\zeta_1=\zeta_2\in[\gamma_2,\gamma_1]$, then the value of  \eqref{eq:totWG} is strictly positive.

In the third and last step of the game, the agents decide whether or not to accept the proposed deal. Including this third step via individual rationality constraints leads to the following welfare gains after bargaining:
\begin{align}WG_i(\zeta_1,\zeta_2)&:=\left\{\begin{array}{ll}\hat{WG}_i(\zeta_1,\zeta_2)&\text{ if }\hat{WG}_1(\zeta_1,\zeta_2)\geq0 \text{ and }\hat{WG}_2(\zeta_1,\zeta_2)\geq0,\\
0&\text{ otherwise,}\end{array}\right. \label{eq:WGi}\end{align}for all $i\in\{1,2\}$.

We next provide an example of the functions $WG_1$ and $WG_2$.

\begin{example}\label{ex:2} Let the preferences be given by $\rho(Y;\gamma_i)=(1-\gamma_i) E[Y]+\gamma_iCVaR_{99\%}(Y)$, $\delta=4/5$, and assume that $X$ is exponentially distributed with parameter 1. Then, we readily get $\rho(X;\gamma_i)= (1-\gamma_i)\cdot 1+\gamma_i(1+F_X^{-1}(99\%))=1+\gamma_i\ln(100)$. So, $\rho(X;\cdot)$ is an affine function. Assume $2/3=\gamma_1> \gamma_2=1/3$. We find that $WG(X)=\rho(X;\gamma_1)-\rho(X;\gamma_2)=\ln(100)/3$. We display the joint strategies where the functions $WG_1$ and $WG_2$ are not both zero in Figure \ref{fig:WG}. In that case, we find
\begin{align*}WG_1(\zeta_1,\zeta_2)&=(\gamma_1-(1-\delta)\zeta_2-\delta\zeta_1)\ln(100)=(\tfrac23-\tfrac15\zeta_2-\tfrac45\zeta_1)\ln(100),\\
WG_2(\zeta_1,\zeta_2)&=(-\gamma_2+(1-\delta)\zeta_2+\delta\zeta_1)\ln(100)=(-\tfrac13+\tfrac15\zeta_2+\tfrac45\zeta_1)\ln(100),
\end{align*} where $(\zeta_1,\zeta_2)\in[0,1]^2$ is in the shaded area in Figure \ref{fig:WG}.
\begin{figure}[h!]\centering
\begin{pspicture}(0,-1)(8,7)
\psline(1,0)(1,6)
\psline(1,6)(7,6)
\psline(7,6)(7,0)\psline(7,0)(1,0)
\uput[0](1,5){$WG_1(\zeta_1,\zeta_2)=WG_2(\zeta_1,\zeta_2)=0$}
\uput[180]{90}(0,3){$\zeta_2\rightarrow$}\uput[-90](4,-.3){$\zeta_1\rightarrow$}
\pspolygon[fillstyle=hlines](3,2)(5,4)(6,0)(3.5,0)
\psline(0.7,0)(1,0)\psline(1,-.3)(1,0)\uput[180](.7,-0.5){0}\psline(0.7,2)(1,2)\uput[180](.7,2){$\gamma_2$}
\psline(0.7,4)(1,4)\uput[180](.7,4){$\gamma_1$}

\psline(0.7,6)(1,6)\uput[180](.7,6){1}
\psline(3,-0.3)(3,0)\uput[-90](3,-.25){$\gamma_2$}\psline(5,-0.3)(5,0)\uput[-90](5,-.25){$\gamma_1$}
\psline(7,-0.3)(7,0)\uput[-90](7,-.25){1}
\end{pspicture}
\caption{Graphical illustration of $WG_i$ as function of $(\zeta_1,\zeta_2)\in[0,1]^2$, corresponding to Example \ref{ex:2}. The functional form is shown in Example \ref{ex:2}. Only in the shaded area, it holds that $WG_i(\zeta_1,\zeta_2)\neq0$ for at least one $i\in\{1,2\}$. Moreover, on this shaded area, $WG_1(\zeta_1,\zeta_2)$ is strictly decreasing in $\zeta_1$ and $\zeta_2$, while $WG_2(\zeta_1,\zeta_2)$ is strictly increasing in $\zeta_1$ and $\zeta_2$. }
\label{fig:WG}\end{figure}
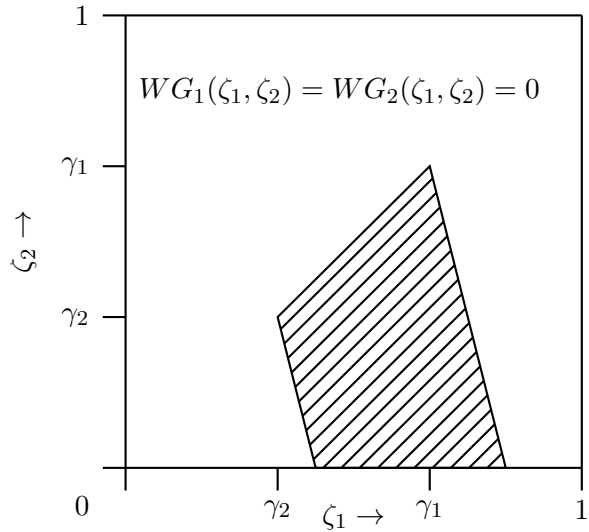
\end{example}

\section{Nash Equilibrium}\label{sec:NE}
Having defined the agents' gains for each pair of strategic choices $(\zeta_1,\zeta_2)$, we are in the position to formulate the associate Nash equilibrium.
\begin{definition} A Nash equilibrium is defined as a pair $(\gamma_1^*,\gamma^*_2)\in[0,1]^2$ such that
\begin{align*}\gamma_1^*&\in \argmax_{\zeta_1\in[0,1]}WG_1(\zeta_1,\gamma^*_2),\\ \gamma^*_2&\in \argmax_{\zeta_2\in[0,1]} WG_2(\gamma^*_1,\zeta_2),\end{align*}
where $WG_1$ and $WG_2$ are defined in \eqref{eq:WGi}.\end{definition}

As stated above, for a Nash equilibrium pair $(\gamma_1^*,\gamma^*_2)\in[0,1]^2$, the induced equilibrium reinsurance contract is given by \eqref{eq:I*} and \eqref{eq:premium}.

Proposition \ref{prop:PO} also implies that under true risk preferences (i.e., when agents do not act strategically), if $\gamma_1<\gamma_2$, then the Pareto-optimal contract $(I^*(X),\pi)=(0,0)$. So, there does not exist a contract in which both the insurer and reinsurer strictly gain more than the status quo. In other words, when the reinsurer's true risk aversion is higher than the risk aversion of the insurer, there is no mutually agreeable transfer of risk. The same holds even when agents strategically choose their risk aversion parameter, as shown in the following result.

\begin{proposition}\label{prop:simpleNE}Suppose that $\rho$ satisfies Assumption \ref{ass:M}. If $\gamma_1\leq \gamma_2$, then any $(\gamma_1^*,\gamma_2^*)\in[0,1]^2$ is a Nash equilibrium.\end{proposition}
\begin{proof}Let $\gamma_1\leq \gamma_2$. If $\gamma_1^*>\gamma^*_2$ or $\gamma^*_1=\gamma^*_2\in[\gamma_2,\gamma_1]$, we get \begin{equation*}\hat{WG}_1(\gamma^*_1,\gamma^*_2)+\hat{WG}_2(\gamma^*_1,\gamma^*_2)=\rho(X;\gamma_1)-\rho(X;\gamma_2)\leq0,\end{equation*} and otherwise $\hat{WG}_1(\gamma^*_1,\gamma^*_2)+\hat{WG}_2(\gamma^*_1,\gamma^*_2)=0+0=0$. Hence, for all $(\gamma^*_1,\gamma^*_2)\in[0,1]^2$ it holds that $\hat{WG}_1(\gamma^*_1,\gamma^*_2)>0$ implies $\hat{WG}_2(\gamma^*_1,\gamma^*_2)<0$, and so we have $WG_1(\gamma^*_1,\gamma^*_2)=0$. Thus, for any $\gamma^*_2\in[0,1]$, it holds that $WG_1(\gamma^*_1,\gamma^*_2)=0$. So, $ \argmax_{\zeta_1\in[0,1]}WG_1(\zeta_1,\gamma^*_2)=[0,1]$. Likewise,  it holds $\argmax_{\zeta_2\in [0,1]}WG_2(\gamma^*_1,\zeta_2)=[0,1]$ for all $\gamma^*_1\in[0,1]$.
Hence, for all $(\gamma^*_1,\gamma^*_2)\in[0,1]^2$, it holds $\gamma^*_1\in \argmax_{\zeta_1\in [0,1]}WG_1(\zeta_1,\gamma^*_2)$ and $\gamma^*_2\in \argmax_{\zeta_2\in [0,1]}WG_2(\gamma^*_1,\zeta_2)$. Hence, any $(\gamma^*_1,\gamma^*_2)\in[0,1]^2$ is a Nash equilibrium.
\end{proof}
Hence, if $\gamma_1\leq \gamma_2$, then the Nash equilibrium leads to a contract where both agents do not improve their position compared to the status quo. In other words, it holds for the Nash equilibrium $(\gamma^*_1,\gamma^*_2)\in[0,1]^2$ that $WG_1(\gamma^*_1,\gamma^*_2)=WG_2(\gamma^*_1,\gamma^*_2)=0$ and a corresponding reinsurance contract is given by $I(X)\stackrel{d}{=}0$ and $\pi=0$: no reinsurance (which is actually the case even without strategic behavior).

Therefore, in the sequel of this paper, we only focus on the case where $\gamma_1> \gamma_2$. 
We first define the value\footnote{We say that $\inf\{\emptyset\}=\infty$ and $\sup\{\emptyset\}=-\infty$.} \begin{equation*}\Gamma_1:=\inf\{\zeta_1\in[0,1]:\hat{WG}_1(\zeta,0)<0\},\end{equation*} which stands for the threshold so that if $\zeta_1>\Gamma_1$, then the insurer will never accept the reinsurance contract, irrespective of the reinsurer's strategy. We then, for any $\zeta_1\in(\gamma_1,\Gamma_1]$, define the function
\begin{align*}f_2(\zeta_1)&=\max\{\zeta_2\in[0,\zeta_1]:\hat{WG}_1(\zeta_1,\zeta_2)\geq0\}\\ &=\max\{\zeta_2\in[0,\zeta_1]:\rho(X;\gamma_1)- (1-\delta)
\rho(X;\zeta_2)-\delta\rho(X;\zeta_1)\geq 0\}.\end{align*}
Note that $f_2$ is the response function on the domain $(\gamma_1,\Gamma_1]$ that sets the insurer at indifference, i.e., $\hat{WG}_1(\zeta_1,f_2(\zeta_1))=0$. This follows from  the fact that $\hat{WG}_1(\zeta_1,\cdot)$ is strictly decreasing and continuous, $\hat{WG}_1(\zeta_1,0)\geq0$ for $\gamma_1^*\leq\Gamma_1$, and $\hat{WG}_1(\zeta_1,\zeta_1)<0$ since $\zeta_1>\gamma_1$.
Since $\rho(X;\cdot)$ is assumed monotone, it holds that $f_2$ is a strictly decreasing function on the domain $(\gamma_1,\Gamma_1]$. Symmetrically, we define
\begin{align*}f_1(\zeta_2)&=\min\{\zeta_1\in[\zeta_2,1]:\hat{WG}_2(\zeta_1,\zeta_2)\geq0\},\end{align*} for all $\zeta_2\in[\Gamma_2,\gamma_2)$, where $\Gamma_2:=\sup\{\zeta_2\in[0,1]:\hat{WG}_2(1,\zeta_2)<0\}$.

The following lemma allows \rev{us} to clarify the values of the agents' best-response functions.

\begin{lemma}\label{lem:BR1} Let $\gamma_1> \gamma_2$, and $\rho$ satisfy Assumption \ref{ass:M}. Then, for all $\zeta_1\in[0,1]$, it holds that
\begin{align}\argmax_{\zeta'_2\in[0,1]}WG_2(\zeta_1,\zeta_2')=\left\{\begin{array}{ll} {[0,1]} &\text{ if }\zeta_1\leq\gamma_2,\\ \zeta_1 &\text{ if }\zeta_1\in(\gamma_2,\gamma_1],\\ f_2(\zeta_1) &\text{ if }\zeta_1\in(\gamma_1,\Gamma_1],\\ {[0,1]} &\text{ if } \zeta_1 > \Gamma_1,\end{array}\right.\label{eq:br2}\end{align} and for all $\zeta_2\in[0,1]$, it holds that
\begin{align}\argmax_{\zeta_1'\in[0,1]}WG_1(\zeta'_1,\zeta_2)=\left\{\begin{array}{ll}{[0,1]} &\text{ if } \zeta_2 < \Gamma_2, \\ f_1(\zeta_2) &\text{ if }\zeta_2\in[\Gamma_2,\gamma_2],\\ \zeta_2 &\text{ if }\zeta_2\in[\gamma_2,\gamma_1),\\ {[0,1]} &\text{ if }\zeta_2\geq\gamma_1.\end{array}\right. \label{eq:br1}\end{align}
\end{lemma}
\begin{proof} We only prove \eqref{eq:br2}, as the proof of \eqref{eq:br1} is similar and thus omitted.
The risk measure $\rho(X;\zeta_2)$ is increasing in $\zeta_2$ by Assumption \ref{ass:M}.  On the domain $[0,\zeta_1)$, we have that $\hat{WG}_2(\zeta_1,\cdot)$ is strictly  increasing. On the domain $(\zeta_1,1]$, we have that $\hat{WG}_2(\zeta_1,\cdot)=0$. Moreover, $\hat{WG}_2(\zeta_1,\zeta_1)=0$ when  $\zeta_1\in[0,1]\backslash[\gamma_2,\gamma_1]$.

We distinguish four cases. In the first case, we let $\zeta_1\in[0,\gamma_2]$. We show that any strategy by agent 2 yields a welfare gain of 0. If $\zeta_2> \zeta_1$, then $(I^*(X),\pi)=(0,0)$ is $(\zeta_1,\zeta_2)$-Pareto optimal, which yields $WG_2(\zeta_1,\zeta_2)=\hat{WG}_2(\zeta_1,\zeta_2)=0$. Thus, it is sufficient to show that $\hat{WG}_2(\zeta_1,\zeta_2)\leq0$ for all $\zeta_2\leq\zeta_1$. This follows from the fact that $\hat{WG}_2(\zeta_1,\cdot)$ is strictly increasing on $[0,\zeta_1]$  and
$$\hat{WG}_2(\zeta_1,\zeta_1)=-\rho(X;\gamma_2)+\rho(X;\zeta_1)\leq -\rho(X;\gamma_2)+\rho(X;\gamma_2)=0,$$ when $\zeta_1\leq \gamma_2$, which follows from the fact that $\rho(X;\gamma)$ is increasing in $\gamma$. Thus, it holds that $WG_2(\zeta_1,\zeta_2)=0$ for all $\zeta_2\in[0,1]$, and thus $ \argmax_{\zeta_2\in[0,1]}WG_2(\zeta_1,\zeta_2)=[0,1]$.

As the second case, we let $\zeta_1\in(\gamma_2,\gamma_1]$. We need to show that $\hat{WG}_1(\zeta_1,\zeta_1)\geq 0$ and $\hat{WG}_2(\zeta_1,\zeta_1)>0$. This follows directly from
\begin{align*}\hat{WG}_1(\zeta_1,\zeta_1)&=\rho(X;\gamma_1)-\rho(X;\zeta_1)\geq \rho(X;\zeta_1)-\rho(X;\zeta_1)=0,
\\ \hat{WG}_2(\zeta_1,\zeta_1)&=-\rho(X;\gamma_2)+\rho(X;\zeta_1)> -\rho(X;\gamma_2)+\rho(X;\gamma_2)=0,\end{align*} which follows from $\zeta_1> \gamma_2$ and the fact that $\rho(X;\gamma)$ is strictly increasing in $\gamma$.

As the third case, we let $\zeta_1\in(\gamma_1,\Gamma_1]$. The function $\hat{WG}_1(\zeta_1,\cdot)=\rho(X;\gamma_1)- (1-\delta)
\rho(X;\cdot)-\delta\rho(X;\zeta_1)$ is strictly decreasing on $(0,\zeta_1)$ and the function $\hat{WG}_2(\zeta_1,\cdot)=-\rho(X;\gamma_2)+(1-\delta)
\rho(X;\cdot)+\delta\rho(X;\zeta_1)$ is strictly increasing on $(0,\zeta_1)$. Agent 2 solves \begin{align*}\max_{\zeta_2\in[0,1]}\hat{WG}_2(\zeta_1,\zeta_2),\text{ s.t. }\hat{WG}_1(\zeta_1,\zeta_2)\geq0,\end{align*} which is thus equal to \begin{align*}\max\{\zeta_2\in[0,\zeta_1]:\hat{WG}_1(\zeta_1,\zeta_2)\geq0\}=f_2(\zeta_1).\end{align*}
We now only need to show that $WG_2(\zeta_1,f_2(\zeta_1))>0$, which implies that there exists a $\zeta_2\in[0,1]$ such that $\hat{WG}_1(\zeta_1,\zeta_2)\geq0$ and $\hat{WG}_2(\zeta_1,\zeta_2)>0$. Since $\zeta_1\leq\Gamma_1$ and since $\hat{WG}_1(\cdot,0)$ is continuous by continuity of $\rho$, it holds that $\hat{WG}_1(\zeta_1,0)\geq0$. Moreover, by monotonicity of $\rho(X;\cdot)$, it follows that \begin{align*}\hat{WG}_2(\zeta_1,0)=-\rho(X;\gamma_2)+(1-\delta)
\rho(X;0)+\delta\rho(X;\zeta_1)>-\rho(X;\gamma_2)+(1-\delta)
\rho(X;\gamma_2)+\delta\rho(X;\gamma_2)=0,\end{align*} which concludes the proof that $\argmax_{\zeta_2\in[0,1]}WG_2(\zeta_1,\zeta_2)=f_2(\zeta_1)$ when $\zeta_1\in(\gamma_1,\Gamma_1]$.

As the fourth and last case, we let $\zeta_1\in(\Gamma_1,1]$. Of course, this case is only relevant when $\Gamma_1<1$. Since $\zeta_1>\Gamma_1$ and since $\hat{WG}_1(\cdot,0)$ is strictly decreasing, it holds that $\hat{WG}_1(\zeta_1,0)<0$. For all $\zeta_2<\zeta_1$, it holds that  \begin{align*}\hat{WG}_1(\zeta_1,\zeta_2)&=\rho(X;\gamma_1)- (1-\delta)
\rho(X;\zeta_2)-\delta\rho(X;\zeta_1)\\ &<\rho(X;\gamma_1)- (1-\delta)
\rho(X;0)-\delta\rho(X;\zeta_1)=\hat{WG}_1(\zeta_1,0)<0.\end{align*} Moreover, for all  $\zeta_2\geq\zeta_1$, $(I^*(X),\pi)=(0,0)$ is $(\zeta_1,\zeta_2)$-Pareto optimal, and thus $\hat{WG}_2(\zeta_1,\zeta_2)=0$. Hence, for all $\zeta_2\in[0,1]$, it holds that $\hat{WG}_1(\zeta_1,\zeta_2)<0$ or $\hat{WG}_2(\zeta_1,\zeta_2)=0$, and thus $WG_2(\zeta_1,\zeta_2)=0$. Hence, we have $\argmax_{\zeta_2\in[0,1]}WG_2(\zeta_1,\zeta_2)=[0,1]$. This concludes the proof.
\end{proof}

If $\zeta_1>\Gamma_1$ and $\zeta_2<\Gamma_2$, then $\hat{WG}_i(\zeta_1,\zeta_2)<0$ for all $i\in\{1,2\}$. This is a violation of Pareto optimality of $(I,\pi)=(X,\pi)$. Thus, at most one of the inequalities $\Gamma_1<1$ and $\Gamma_2>0$ hold. An example of the best-response correspondences in Lemma \ref{lem:BR1} is provided in Figure {\ref{fig:exampleBR}.

\begin{figure}[h!]\centering
\begin{pspicture}(0,-1)(17,7)
\psline(1,0)(1,6)
\psline(1,6)(7,6)
\psline(7,6)(7,0)\psline(7,0)(1,0)
\psline(3,0)(3,6)
\psline[linestyle=dotted](6,0)(6,6)
\psline(3,2)(5,4)\psline(5,4)(6,0)
\uput[180]{90}(0,3){$\argmax_{\zeta_2'\in[0,1]}WG_2(\zeta_1,\zeta_2')\rightarrow$}
\uput[-90](4,-.3){$\zeta_1\rightarrow$}
\psframe[fillstyle=hlines,linestyle=none](1,0)(3,6)
\psframe[fillstyle=hlines,linestyle=none](6,0)(7,6)
\psline(0.7,0)(1,0)\psline(1,-.3)(1,0)\uput[180](.7,-0.5){0}\psline(0.7,2)(1,2)\uput[180](.7,2){$\gamma_2$}
\psline(0.7,4)(1,4)\uput[180](.7,4){$\gamma_1$}
\psline(6,-.3)(6,0)\uput[-90](6,-.25){$\Gamma_1$}
\psline(0.7,6)(1,6)\uput[180](.7,6){1}
\psline(3,-0.3)(3,0)\uput[-90](3,-.25){$\gamma_2$}\psline(5,-0.3)(5,0)\uput[-90](5,-.25){$\gamma_1$}
\psline(7,-0.3)(7,0)\uput[-90](7,-.25){1}

\psline(10,0)(10,6)
\psline(10,6)(16,6)
\psline(16,6)(16,0)\psline(16,0)(10,0)
\psline(14,0)(14,6)
\psline(12,2)(14,4)\psline(12,2)(10,3)
\uput[180]{90}(9,3){$\argmax_{\zeta_1'\in[0,1]}WG_1(\zeta'_1,\zeta_2)\rightarrow$}
\uput[-90](13,-.3){$\zeta_2\rightarrow$}
\psframe[fillstyle=hlines,linestyle=none](14,0)(16,6)
\psline(9.7,0)(10,0)\psline(10,-.3)(10,0)\uput[180](9.7,-0.5){0}\psline(9.7,2)(10,2)\uput[180](9.7,2){$\gamma_2$}
\psline(9.7,4)(10,4)\uput[180](9.7,4){$\gamma_1$}
\psline(9.7,6)(10,6)\uput[180](9.7,6){1}
\psline(12,-0.3)(12,0)\uput[-90](12,-.25){$\gamma_2$}\psline(14,-0.3)(14,0)\uput[-90](14,-.25){$\gamma_1$}
\psline(16,-0.3)(16,0)\uput[-90](16,-.25){1}
\end{pspicture}
\caption{Graphical illustration of the best-response correspondences $\argmax_{\zeta_2'\in[0,1]}WG_2(\zeta_1,\zeta_2')$ (left graph) and $\argmax_{\zeta_1'\in[0,1]}WG_1(\zeta'_1,\zeta_2)$ (right graph), corresponding to the setting in Example \ref{ex:2}. The best-respondences are plotted against the parameter $\zeta_i\in[0,1]$ of the other agent. The functional form is shown in \eqref{eq:br2} and \eqref{eq:br1}. Here, we find $\Gamma_1=5/6$, and $\Gamma_2=-\infty$. 
A solid line means that the line itself constitutes the best-responses, while a dotted line means that the line itself does not constitutes the best-responses.}
\label{fig:exampleBR}\end{figure}
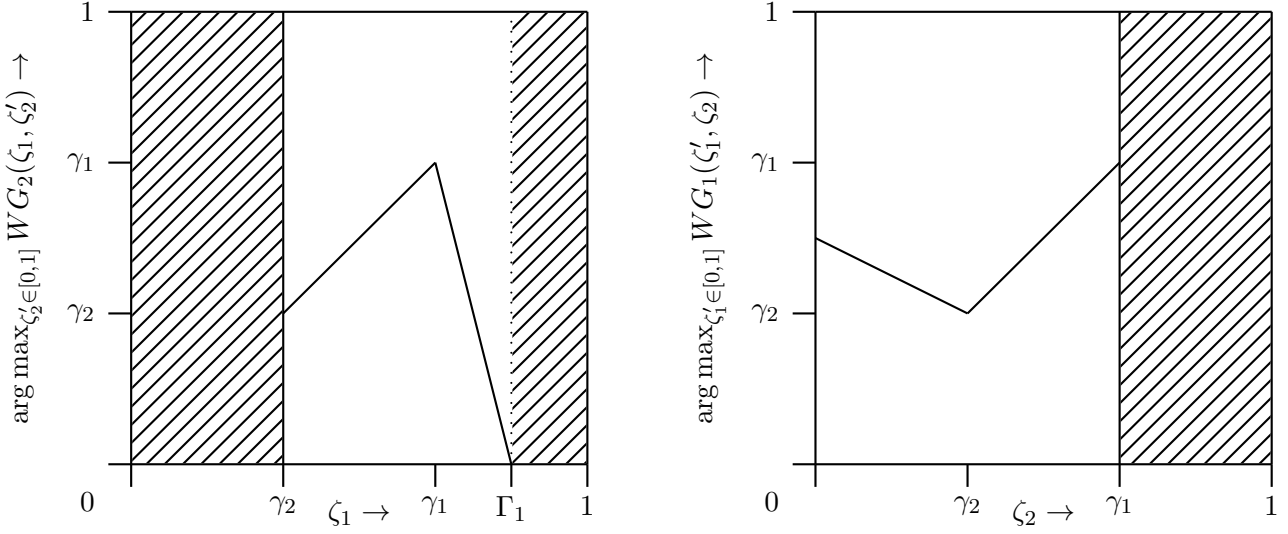

Based on \rev{Lemma \ref{lem:BR1}}, we can identify \rev{a set of Nash equilibria}.

\begin{theorem}\label{th:NE}Let $\gamma_1> \gamma_2$, and $\rho$ satisfy Assumption \ref{ass:M}. A set of Nash equilibria is given by
\begin{equation}\left\{(\gamma_1^*,\gamma_2^*):\gamma_1^*\in[0,\gamma_2]\cup (\Gamma_1,1],\gamma_2^*\in[0,\Gamma_2)\cup[\gamma_1,0]\right\} \cup\{(\gamma,\gamma):\gamma\in[\gamma_2,\gamma_1]\}.\label{eq:SNE}\end{equation}
\end{theorem}
\begin{proof} From \eqref{eq:br2} and \eqref{eq:br1}, we immediately get that $\gamma\in \argmax_{\zeta_1\in [0,1]}WG_1(\zeta_1,\gamma)$ and $\gamma\in \argmax_{\zeta_2\in [0,1]}WG_2(\gamma,\zeta_2)$ for all $\gamma\in[\gamma_2,\gamma_1]$. Thus, $(\gamma,\gamma)$ is a Nash equilibrium for any $\gamma\in[\gamma_2,\gamma_1]$.
If $\gamma_1^*\in[0,\gamma_2]\cup (\Gamma_1,1]$ and $\gamma_2^*\in[0,\Gamma_2)\cup[\gamma_1,0]$, then we get again from Lemma \ref{lem:BR1} that $\argmax_{\zeta_2\in[0,1]}WG_2(\gamma^*_1,\zeta_2)=\argmax_{\zeta_1\in[0,1]}WG_1(\zeta_1,\gamma_2^*)=[0,1]$. Therefore, we have that $\gamma_2^*\in\argmax_{\zeta_2\in[0,1]}WG_2(\gamma^*_1,\zeta_2)$ and $\gamma_1^*\in\argmax_{\zeta_1\in[0,1]}WG_1(\zeta_1,\gamma_2^*)$, and thus $(\gamma_1^*,\gamma_2^*)$ is a Nash equilibrium. So, the set \eqref{eq:SNE} is a set of Nash equilibria.
\end{proof}

We display an example of Nash equilibria in Figure \ref{fig:1}.
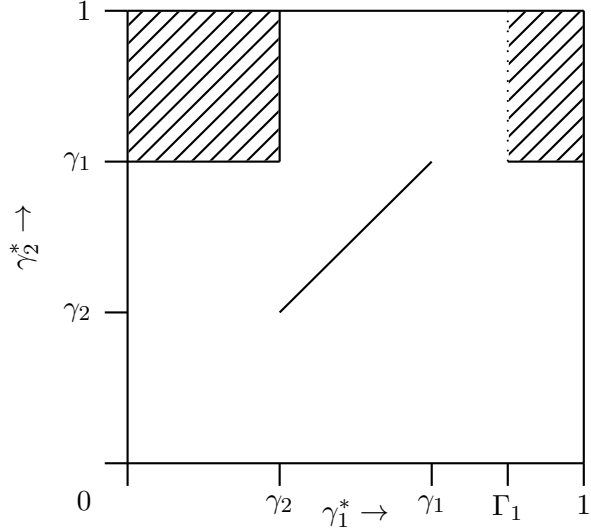
\begin{figure}[h!]\centering
\begin{pspicture}(0,-1)(8,7)
\psline(1,0)(1,6)
\psline(1,6)(7,6)
\psline(7,6)(7,0)\psline(7,0)(1,0)
\psline(1,4)(3,4)\psline(3,4)(3,6)
\psline[linestyle=dotted](6,4)(6,6)\psline(6,4)(7,4)
\psline(3,2)(5,4)
\uput[180]{90}(0,3){$\gamma_2^*\rightarrow$}\uput[-90](4,-.3){$\gamma_1^*\rightarrow$}
\psframe[fillstyle=hlines,linestyle=none](1,4)(3,6)
\psframe[fillstyle=hlines,linestyle=none](6,4)(7,6)
\psline(0.7,0)(1,0)\psline(1,-.3)(1,0)\uput[180](.7,-0.5){0}\psline(0.7,2)(1,2)\uput[180](.7,2){$\gamma_2$}
\psline(0.7,4)(1,4)\uput[180](.7,4){$\gamma_1$}
\psline(6,-.3)(6,0)\uput[-90](6,-.25){$\Gamma_1$}
\psline(0.7,6)(1,6)\uput[180](.7,6){1}
\psline(3,-0.3)(3,0)\uput[-90](3,-.25){$\gamma_2$}\psline(5,-0.3)(5,0)\uput[-90](5,-.25){$\gamma_1$}
\psline(7,-0.3)(7,0)\uput[-90](7,-.25){1}
\end{pspicture}
\caption{Graphical illustration of Nash equilibria if $2/3=\gamma_1> \gamma_2=1/3$, $\Gamma_1=5/6$, and $\Gamma_2=-\infty$. These Nash equilibria are given in \eqref{eq:SNE}. A solid line means that the line itself constitutes Nash equilibria, while a dotted line means that the line itself does not constitutes Nash equilibria.}
\label{fig:1}\end{figure}

We get directly from the comonotonic additivity that for any $\gamma\in[\gamma_2,\gamma_1]$, the allocation of the welfare gains between the insurer and the reinsurer is
\begin{eqnarray} WG_1(\gamma,\gamma)&=&\hat{WG}_1(\gamma,\gamma)=\rho(X;\gamma_1)-\rho(X;\gamma)\geq 0,\label{wgg1}\\
 WG_2(\gamma,\gamma)&=&\hat{WG}_2(\gamma,\gamma)=\rho(X;\gamma)-\rho(X;\gamma_2)\geq 0,\label{wgg2} \end{eqnarray}
It follows that for $\gamma\in(\gamma_2,\gamma_1)$, the Nash equilibrium $(\gamma,\gamma)$ induces a reinsurance contract that yields a strict improvement compared to the status quo for both agents. Theorem \ref{th:NE} selects some Nash equilibria, but not necessarily all of them. The next result shows that for all Nash equilibria not in \eqref{eq:SNE}, both agents are indifferent compared to the status quo, in other words we determine all Nash equilibria in which at least one agent strictly benefits. In fact, we show that at equilibrium agents appear as having the same risk aversion parameters.

\begin{theorem} Let $\gamma_1> \gamma_2$, and $\rho$ satisfy Assumption \ref{ass:M}. The set of Nash equilibria $(\gamma_1^*,\gamma_2^*)$ in which there exists an $i\in\{1,2\}$ for which $WG_i(\gamma_1^*,\gamma_2^*)>0$ is given by
\begin{equation}\{(\gamma,\gamma):\gamma\in[\gamma_2,\gamma_1]\}.\label{eq:SNEP}\end{equation}\end{theorem}\begin{proof}
From Theorem \ref{th:NE} it follows that the set \eqref{eq:SNEP} is a subset of the set of Nash equilibria. From \eqref{wgg1}-\eqref{wgg2} and Assumption \ref{ass:M}, it follows that all elements of the set \eqref{eq:SNEP} yield a contract for which at least one agent has a strictly positive welfare gain. So, in this proof, we only need to show that there does not exist another Nash equilibrium $(\gamma_1^*,\gamma_2^*)$ in which there exists an $i\in\{1,2\}$ for which $WG_i(\gamma_1^*,\gamma_2^*)>0$.

Suppose that $(\gamma_1^*,\gamma_2^*)\in[0,1]^2$ is a Nash equilibrium that is not in \eqref{eq:SNEP}, but with $WG_2(\gamma_1^*,\gamma_2^*)>0$. It then holds that $\hat{WG}_1(\gamma_1^*,\gamma_2^*)\geq0$ and $\hat{WG}_2(\gamma_1^*,\gamma_2^*)>0$. Then, we get from the definition of $\hat{WG}_2$ that $\gamma_1^*>\gamma_2^*$  or $\gamma_1^*=\gamma_2^*\in[\gamma_2,\gamma_1]$. Since $(\gamma_1^*,\gamma_2^*)$ is not in \eqref{eq:SNEP}, it must hold $\gamma_1^*>\gamma_2^*$. Since $\hat{WG}_1(\cdot,\gamma_2^*)=\rho(X;\gamma_1)- (1-\delta)
\rho(X;\gamma_2^*)-\delta\rho(X;\cdot)$ is strictly decreasing on and $\hat{WG}_2(\cdot,\gamma^*_2)=-\rho(X;\gamma_2)+(1-\delta)
\rho(X;\gamma_2^*)+\delta\rho(X;\cdot)$ is continuous on $(\gamma_2^*,1]$ by Assumption \ref{ass:M}, there exists a $\zeta_1\in(\gamma_2^*,\gamma_1^*)$ such that $\hat{WG}_1(\zeta_1,\gamma_2^*)>\hat{WG}_1(\gamma_1^*,\gamma_2^*)\geq0$ and $\hat{WG}_2(\zeta_1,\gamma_2^*)\geq0$. So, $WG_1(\zeta_1,\gamma_2^*)>\hat{WG}_1(\gamma_1^*,\gamma_2^*)=WG_1(\gamma_1^*,\gamma_2^*)$. Hence, $\zeta_1$ is a better response than $\gamma_1^*$ for agent 1, which is a contradiction with the assumption that $(\gamma_1^*,\gamma_2^*)$ is a Nash equilibrium.

The proof that there does not exist a Nash equilibrium $(\gamma_1^*,\gamma_2^*)\in[0,1]^2$ with $WG_1(\gamma_1^*,\gamma_2^*)>0$ which is not in \eqref{eq:SNEP} is similar, and thus omitted.
\end{proof}

Note that for the Nash equilibria $(\gamma,\gamma)$ with $\gamma\in[\gamma_2,\gamma_1]$, the corresponding reinsurance contract $(I,\pi)$ is $(\gamma_1,\gamma_2)$-Pareto optimal (i.e., it coincides with the Pareto optimal without the strategic behavior). The latter follows from the fact that $I$ solves \eqref{eq:I} \citep[see Theorem 3.1 of][]{asimit2018}. Hence, if there is some reinsurance contract that strictly benefits at least one of the agents (which holds when $\gamma_1>\gamma_2$), we have shown that \eqref{eq:SNEP} is the set of Pareto optimal Nash equilibria.  The exact choice of the common risk parameter at the equilibrium does not affect the contract but only the premium and hence the individual welfare gains, or in other words the allocation of the total welfare gain between them (larger $\gamma$ means higher premium paid by the insurer and hence more gain for the reinsurer).

\begin{remark}\label{r:no delta}
It is important to emphasize that under the Nash equilibria $(\gamma,\gamma)$ with $\gamma\in[\gamma_2,\gamma_1]$, the exogenous choice of $\delta$ is irrelevant for the reinsurance contract $(I,\pi)$. In other words, the Nash equilibrium game on the risk aversion coefficient transfers the market power to the exact choice of the common submitted risk aversion $\gamma$. This is a significant feature of the proposed equilibrium model, since the agent's market power is now determined within the equilibrium and not as an exogenous parameter.
\end{remark}
\rev{\begin{remark}For any Nash equilibrium $(\gamma,\gamma)$ with $\gamma\in[\gamma_2,\gamma_1]$, it holds that $\pi=\rho(I(X);\gamma)$. Since $\gamma\geq\gamma_2$, it follows that $\pi\geq\rho(I(X);\gamma_2)$. If $\rho$ is a distortion risk measure with $g(s;\gamma_2)\geq s$ for all $s\in[0,1]$, then it holds that $\pi\geq\rho(I(X);\gamma_2)\geq E[I(X)]$.
\end{remark}}

\subsection{Stackelberg equilibrium}

In the discussion above, we showed that even the Nash equilibria that strictly increase the agents' welfare are usually non-unique. One way to select a particular one could be through the concept of the Stackelberg equilibrium. It turns out that in our setting the Stackelberg equilibrium is a refinement of the Nash equilibrium.

More precisely, we have seen that the Nash equilibrium which yields improvement of the status quo is when agents appear homogeneous with respect to their submitted risk aversions. We also have seen that there is  monotonicity of the gains' allocation with respect to the common equilibrium risk aversion (higher $\gamma$, higher gain to the reinsurer). Therefore, the determination of a unique Nash equilibrium is equivalent to the determination of the common $\gamma$.

According to the Stackelberg equilibrium, one of the two agents moves first (called the \emph{leader}), and submits his parameter within the strategic set. Thereafter, the other agent (called the \emph{follower}) observes this parameter, and responds by stating his best response in an optimal way. In other words, the leader of the game decides on an optimal strategy while taking into consideration the optimal response strategies of the follower\footnote{Similar equilibrium argument has been applied in \cite{chen2018}, where the role of the leader is given to the reinsurer.}.   Below, we state the formal definition of the Stackelberg equilibrium within our Nash equilibrium setting.

\begin{definition} The strategy profile $(\gamma_1^*,\gamma_2^*)\in[0,1]^2$ is a Stackelberg equilibrium when: \begin{itemize}\item if the insurer is the leader and the reinsurer the follower, then
$$\gamma_2^*\in \argmax_{\zeta_2\in[0,1]}WG_2(\gamma^*_1,\zeta_2),$$
and there is no other pair $(\hat\gamma_1,\hat \gamma_2)$ such that
$$\hat \gamma_2\in \argmax_{\zeta_2\in[0,1]}WG_2(\hat \gamma_1,\zeta_2)\text{  and } WG_1(\gamma^*_1,\gamma^*_2)<WG_1(\hat\gamma_1,\hat\gamma_2);$$
\item if the reinsurer is the leader and the insurer the follower, then
$$\gamma_1^*\in \argmax_{\zeta_1\in[0,1]}WG_1(\zeta_1,\gamma_2^*),$$ and there is no other pair $(\hat\gamma_1,\hat \gamma_2)$ such that
$$\hat \gamma_1\in \argmax_{\zeta_1\in[0,1]}WG_1( \zeta_1,\hat\gamma_2)\text{ and }WG_2(\gamma^*_1,\gamma^*_2)<WG_2(\hat\gamma_1,\hat\gamma_2).$$
    \end{itemize}
\end{definition}
We have seen in Theorem \ref{th:Stack} that the Nash equilibrium with strict improvement is the one that each agent mimics the strategy of the other. If the leader of the transaction takes this into account, the following lemma describes his optimal decision.
 \begin{lemma}\label{lem:Stack} If $\gamma_1> \gamma_2$, the following statements hold:
\begin{itemize}\item  $\argmax_{\gamma\in [0,1]}WG_2(\gamma,\gamma)$ is single-valued and given by $\gamma_1$;
\item $\argmax_{\gamma\in [0,1]}WG_1(\gamma,\gamma)$ is single-valued and given by $\gamma_2$.\end{itemize}\end{lemma}

\begin{proof} We show the first result, and the proof of the second result is similar and thus omitted. If $\gamma< \gamma_2$ then $\hat{WG}_2(\gamma,\gamma)<0$, and if $\gamma> \gamma_1$ then $\hat{WG}_1(\gamma,\gamma)<0$. So, since $WG_2(\gamma,\gamma)\geq0$ for all $\gamma\in[0,1]$, we get by construction that $$\max_{\gamma\in [0,1]}WG_2(\gamma,\gamma)=\max_{\gamma\in [\gamma_2,\gamma_1]}WG_2(\gamma,\gamma).$$ From \eqref{wgg1}-\eqref{wgg2} and monotonicity of $\rho(X;\cdot)$, it follows that $\hat{WG}_1(\gamma,\gamma)\geq0$ and $\hat{WG}_2(\gamma,\gamma)\geq0$ for all $\gamma\in [\gamma_2,\gamma_1]$, and thus
$$\max_{\gamma\in [\gamma_2,\gamma_1]}WG_2(\gamma,\gamma)=\max_{\gamma\in [\gamma_2,\gamma_1]}\hat{WG}_2(\gamma,\gamma)=\max_{\gamma\in [\gamma_2,\gamma_1]}-\rho(X;\gamma_2)+\rho(X;\gamma).$$
By monotonicity of $\rho(X;\cdot)$, there is only one maximizer, and in fact equal to $\gamma_1$. Therefore, $\argmax_{\gamma\in [0,1]}WG_2(\gamma,\gamma)$ is single-valued and given by $\gamma_1$. This concludes the proof.
\end{proof}

A description of the Stackelberg equilibria is stated in the following theorem, where Lemma \ref{lem:Stack} is used.
\begin{theorem}\label{th:Stack}Let $\gamma_1> \gamma_2$. Then, the following statements hold:
\begin{itemize}\item if the insurer is the leader,
the Stackelberg equilibrium is given by $(\gamma_2,\gamma_2)$ and all welfare gain goes to the insurer;
\item if the reinsurer is the leader, the Stackelberg equilibrium is given by $(\gamma_1,\gamma_1)$ and all the welfare gain goes to the reinsurer.\end{itemize}
\end{theorem}
\begin{proof} Let $\gamma_1> \gamma_2$, and first the insurer be the leader and the reinsurer be the follower. Then, it follows from Assumption \ref{ass:M} and \eqref{eq:totWG} that
$\hat{WG_1}(\gamma_1^*,\gamma_2^*)+\hat{WG}_2(\gamma_1^*,\gamma_2^*)$ is positive, and the same for all $\gamma_1^*>\gamma_2^*$ or $\gamma^*_1=\gamma^*_2\in[\gamma_2,\gamma_1]$. Moreover, we readily verify that $\hat WG_2(\gamma_2,\gamma_2)=0$ and $\gamma_2\in \argmax_{\zeta_2\in[0,1]}WG_2(\gamma_2,\zeta_2)$, and thus $(\gamma_2,\gamma_2)$ yields a Stackelberg equilibrium with $WG_1(\gamma_2,\gamma_2)>0$. The question now is whether there are more Stackelberg equilibria. In order words, does there exist $(\gamma_1,\gamma_2^*)\in[0,1]^2$ such that  $\gamma_1^*>\gamma_2^*$ or $\gamma^*_1=\gamma^*_2\in[\gamma_2,\gamma_1]$ with $\hat WG_2(\gamma_1^*,\gamma_2^*)=0$ and $\gamma_2^*\in \argmax_{\zeta_2\in[0,1]}WG_2(\gamma_1^*,\zeta_2)$?

We distinguish five different choices of $\gamma_1^*\in[0,1]$. First, if $\gamma_1^*\in[0,\gamma_2)$, then we showed in the proof of Lemma \ref{lem:BR1} that $\hat{WG}_2(\gamma_1^*,\gamma_2^*)<0$ for all $\gamma_2^*\in[0,\gamma_1^*]$.
Thus, there cannot exist a Stackelberg equilibrium $(\gamma_1^*,\gamma_2^*)$ with $\gamma_1^*\in[0,\gamma_2)$.

Second, if $\gamma_1^*=\gamma_2$, it follows from the fact that $\hat{WG}_2(\gamma_2,\cdot)$ is strictly increasing on $[0,\gamma_2]$ and
$\hat{WG}_2(\gamma_2,\gamma_2)=0,$ that $\hat{WG}_2(\gamma_2,\gamma_2^*)<0$ for all $\gamma_2^*\in[0,\gamma_2)$. Moreover, by definition, it holds that $WG_1(\gamma_2,\gamma_2^*)=0$ for all $\gamma_2^*\in(\gamma_2,1]$. Thus, the only Stackelberg equilibrium $(\gamma_1^*,\gamma_2^*)$ with $\gamma_1^*=\gamma_2$ is $(\gamma_2,\gamma_2)$.

Third, if $\gamma_1^*\in(\gamma_2,\gamma_1]$, then we get from Lemma \ref{lem:BR1} that $\argmax_{\zeta_2\in[0,1]}WG_2(\gamma_1^*,\zeta_2)=\gamma_1^*$. From Lemma \ref{lem:Stack}, we get that $WG_1(\gamma^*_1,\gamma_1^*)<WG_1(\gamma_2,\gamma_2)$, and thus the insurer better plays the strategy $\gamma_1^*=\gamma_2$, which implies that there cannot exist a Stackelberg equilibrium $(\gamma_1^*,\gamma_2^*)$ with $\gamma_1^*\in(\gamma_2,\gamma_1]$.

Fourth, if $\gamma_1^*\in(\gamma_1,\Gamma_1]$, then we get from Lemma \ref{lem:BR1} that $\argmax_{\zeta_2\in[0,1]}WG_2(\gamma_1^*,\zeta_2)=f_2(\gamma_1^*)$. This implies that
$\hat{WG}_1(\gamma_1^*,f_2(\gamma_1^*))=0<WG_1(\gamma_2,\gamma_2)$. Hence, the insurer chooses the strategy $\gamma_1^*=\gamma_2$, which implies that there cannot exist a Stackelberg equilibrium $(\gamma_1^*,\gamma_2^*)$ with $\gamma_1^*\in[\gamma_1,\Gamma_1]$.

Fifth and last, if $\gamma_1^*\in(\Gamma_1,1]$,  then we showed in the proof of Lemma \ref{lem:BR1} that $\hat{WG}_1(\gamma_1^*,\gamma_2^*)\leq \hat{WG}_1(\gamma_1^*,0)< 0$ for all $\gamma_2^*\in[0,\gamma_1^*)$.
Thus, there cannot exist a Stackelberg equilibrium $(\gamma_1^*,\gamma_2^*)$ with $\gamma_1^*\in(\Gamma_1,1]$. This concludes the proof that the only Stackelberg equilibrium is $(\gamma_2,\gamma_2)$ when  the insurer is the leader and the reinsurer is the follower. The proof of the second result is similar, and thus omitted.
\end{proof}
We may now sum up how the situation is formed under the Stackelberg equilibrium argument. We first recall that if $\gamma_1\leq \gamma_2$, there is no contract that yields Pareto improvement to the status quo. For all $(\gamma_1^*,\gamma_2^*)\in[0,1]^2$, we have that $WG_i(\gamma_1^*,\gamma_2^*)=0$, and thus every  $(\gamma_1^*,\gamma_2^*)\in[0,1]^2$ constitutes a Stackelberg equilibrium. On the other hand, if $\gamma_1> \gamma_2$, then we get from \eqref{eq:premium}, \eqref{eq:I*}, and Theorem \ref{th:Stack} that there is full insurance in the Stackelberg equilibrium: $I^*(X)=X$, with the premium given by
\begin{equation*}\pi=\left\{\begin{array}{ll}\rho(X;\gamma_2) &\text{if the insurer is the leader},\\ \rho(X;\gamma_1) &\text{if the reinsurer is the leader}.\end{array}\right.
\end{equation*}
This constitutes a Nash equilibrium, where the leader extracts all welfare gain of the risk-transfer making the follower indifferent between this transaction and the status quo (no reinsurance). Indeed, we readily get that under that premium, the welfare gain of the leader  is $\rho(X;\gamma_1)-\rho(X;\gamma_2)$ and zero for the follower.

Taking into account the discussion after Theorem \ref{th:NE}, we conclude that the unique Stackelberg equilibrium could be seen as an extreme one, in the sense that all the gain is given to the leader. We emphasize again that the parameter $\delta$ does not affect the transaction anymore (recall Remark \ref{r:no delta}), which means that under Stackelberg equilibrium all the market power is transferred to the leader, even if his exogenously given bargaining power $\delta$ or $1-\delta$ is low. This is an important coincidence of the strategic behavior with respect to the agents' risk aversion, which could heavily change the induced welfare gains.

\subsection{On welfare gains}\label{sec:gains comparison}

Since agents play a game for the determination of the reinsurance contract, it is expected that the total welfare gain is less when compared with the non-strategic situation. It turns our however, that this is not the case in our game, which means that under social welfare terms, the Nash equilibrium is again optimal. The only feature that heavily changes is the premium and hence the agents' individual welfare gain.

More precisely, when agents do not act strategically, the Pareto-optimal contract is given by Proposition \ref{prop:PO}, while the premium $\pi$ by \eqref{eq:premium} (with $\zeta_i=\gamma_i$). In particular, when $\gamma_1<\gamma_2$ there is no reinsurance, while when $\gamma_1\geq \gamma_2$ we readily calculate that the corresponding welfare (optimal) gains are given by\begin{align*}
WG_1^o &=(1-\delta)(\rho(X;\gamma_1)-\rho(X;\gamma_2)), \\
WG_2^o &=\delta(\rho(X;\gamma_1)-\rho(X;\gamma_2)).
\end{align*}
Note that in contrast to the Nash equilibrium, the individual gains do depend clearly on the bargaining power $\delta$.

On the other hand, we have seen that the strictly beneficial Nash equilibria are the ones with a common equilibrium risk parameter $\gamma^*\in[\gamma_2,\gamma_1]$ and the individual welfare gains become
\begin{align*}
WG_1(\gamma^*,\gamma^*) &=\rho(X;\gamma_1)-\rho(X;\gamma^*), \\
WG_2(\gamma^*,\gamma^*) &=\rho(X;\gamma^*)-\rho(X;\gamma_2).
\end{align*}
Therefore, although the total welfare after the transaction is the same with and without the game on risk preferences, its allocation between agents may change dramatically.

Note that $\gamma^*\in[\gamma_2,\gamma_1]$ can be chosen such that $WG_i^o=WG_i(\gamma^*,\gamma^*)$ for all $i\in\{1,2\}$, i.e., both equilibria will be the same, however this is a very special case. This is because the underlying concept of $\delta$ as bargaining power of the reinsurer does not work in the Nash equilibrium game, where the agent who wins more after the transaction depends on who has a better position in choosing $\gamma^*$. As mentioned before, Stackelberg equilibrium is just an extreme case, where the leader gets it all.

\begin{example}
We return to Example \ref{ex:2}. Recall $WG(X)=\ln(100)/3$. We derive that {\allowdisplaybreaks\begin{align*}
WG_1(\gamma^*,\gamma^*) &=\frac{\gamma_1-\gamma^*}{\gamma_1-\gamma_2}WG(X)=(\tfrac23-\gamma^*)\ln(100), \\
WG_2(\gamma^*,\gamma^*) &= \frac{\gamma^*-\gamma_2}{\gamma_1-\gamma_2}WG(X)=(\gamma^*-\tfrac13)\ln(100).
\end{align*}}
Note that only when $\delta=(\gamma^*-\gamma_2)/(\gamma_1-\gamma_2)=3\gamma^*-1$, (or equivalent when $\gamma^*$ is chosen such that $\gamma^*=(1+\delta)/3$), we have $WG_i^o=WG_i(\gamma^*,\gamma^*)$ for $i\in\{1,2\}$.
\end{example}

\subsection{A remark on strict monotonicity of $\rho$}\label{subsec:on strict monotonicity}
We recall that in Assumption \ref{ass:M}, we ask the risk measure $\rho$ to be \textit{strictly} increasing in $\gamma$. Suppose, we only require $\rho(I(X);\cdot)$ to be increasing, which holds for instance when $$\rho(I(X);\gamma)=VaR_{\gamma}(I(X)):=\inf\{x\in\mathbb{R}:\mathbb{P}(I(X)\leq x)\geq\gamma\}$$ or $\rho(I(X);\gamma)=CVaR_{\gamma}(I(X))$, where $I\in\mathcal I$, $\gamma\in[0,1]$, and where $CVaR_{\alpha}(I(X))$ is defined in Example \ref{ex:1}. Then, all insurance contracts stated in Proposition \ref{prop:PO} are still Pareto optimal, but there may be more Pareto optimal contracts. As a result, the correspondence in Lemma \ref{lem:BR1} is now only a subset of the best-response correspondence. As a result, the Nash equilibria stated in Proposition \ref{prop:simpleNE} and Theorem \ref{th:NE} are still Nash equilibria. Moreover, the set in \eqref{eq:SNEP} is now a set of Nash equilibria where at least one agent profits compared to the status quo, but it is not necessarily consisting of all such Nash equilibria. Likewise, a set of Stackelberg equilibria is given in Theorem \ref{th:Stack}, but there may be more. To summarize, all results of this paper except the uniqueness would hold true.

\section{Conclusion}\label{sec:conclusion}

This paper introduces a strategic behavior in reinsurance transactions that is based on the argument that both agents have \rev{the motive} to strategically choose the risk preferences that they will appear to have in the transaction. Following the related literature on thin financial markets and non-competitive risk-sharing, we formulate a game according to which agents agree to apply the Pareto-optimal sharing rule for any risk preferences they submit. For the agents' strategic set, we impose a parameterization of the (monotone and comonotonic additive) risk measures that allows us to interpret the agents' choice as a risk-aversion coefficient. In fact, almost all the well-known risk measures are consistent with this parameterization.

After proving the well-posedness of the best-response problem, we show that, at the strictly beneficial Nash equilibrium, agents appear homogeneous with respect to their risk aversions. More importantly, there is no loss of total welfare caused by the game, and only the allocation of welfare gains is affected by agents' strategic behavior. It is furthermore shown that the gain allocation does not depend on agents' (potentially asymmetric) bargaining power, whose influence vanishes through the game procedure.

The exact premium (and hence the allocation of gains) is determined by the choice of the common risk aversion of the equilibrium, where higher value means higher gain for the reinsurer. This needs an extra criterion, and one already used in the related literature, is based on the notion of Stackelberg equilibrium. \rev{In a Stackelberg equilibrium}, one of the agents is the leader who steps first on the best-response procedure and the other follows. It turns out that in our game, the leader is the one that gets all the welfare gain, leaving the follower indifferent between the status quo and the equilibrium reinsurance.

\bigskip

\bibliographystyle{chicago}
\bibliography{Bib_MichailTim}

\end{document}